\documentclass[nohyperref]{article}

\usepackage{microtype}
\usepackage{graphicx}
\usepackage{subfigure}
\usepackage{booktabs} % for professional tables
\usepackage{balance}
\usepackage{hyperref}

\usepackage{csquotes}
\usepackage{bm}
\usepackage{multirow}
\usepackage[accepted]{icml2022}

\usepackage{amsmath}
\usepackage{amssymb}
\usepackage{mathtools}
\usepackage{amsthm}

\usepackage[capitalize,noabbrev]{cleveref}

\theoremstyle{plain}
\newtheorem{theorem}{Theorem}[section]
\newtheorem{proposition}[theorem]{Proposition}

\theoremstyle{definition}

\theoremstyle{remark}

\usepackage[textsize=tiny]{todonotes}

\icmltitlerunning{A Unified Perspective on Deep Equilibrium Finding}

\begin{document}

\twocolumn[
\icmltitle{A Unified Perspective on Deep Equilibrium Finding}



\icmlsetsymbol{equal}{*}

\begin{icmlauthorlist}
\icmlauthor{Xinrun Wang}{ntu}
\icmlauthor{Jakub \v{C}ern\'{y}}{equal,ntu}
\icmlauthor{Shuxin Li}{equal,ntu}
\icmlauthor{Chang Yang}{equal,nus}
\icmlauthor{Zhuyun Yin}{ntu}
\icmlauthor{Hau Chan}{unl}
\icmlauthor{Bo An}{ntu}
\end{icmlauthorlist}

\icmlaffiliation{ntu}{Nanyang Technological University}
\icmlaffiliation{nus}{National University of Singapore}
\icmlaffiliation{unl}{University of Nebraska-Lincoln}

\icmlcorrespondingauthor{Xinrun Wang}{xinrun.wang@ntu.edu.sg}

\icmlkeywords{Machine Learning, ICML}

\vskip 0.3in
]



\printAffiliationsAndNotice{\icmlEqualContribution} % otherwise use the standard text.

\begin{abstract}
Extensive-form games provide a versatile framework for modeling interactions of multiple agents subjected to imperfect observations and stochastic events. In recent years, two paradigms, policy space response oracles (PSRO) and counterfactual regret minimization (CFR), showed that extensive-form games may indeed be solved efficiently. Both of them are capable of leveraging deep neural networks to tackle the scalability issues inherent to extensive-form games and we refer to them as deep equilibrium-finding algorithms. Even though PSRO and CFR share some similarities, they are often regarded as distinct and the answer to the question of which is superior to the other remains ambiguous. Instead of answering this question directly, in this work we propose a unified perspective on deep equilibrium finding that generalizes both PSRO and CFR. Our four main contributions include: i) a novel response oracle (RO) which computes Q values as well as reaching probability values and baseline values; ii) two transform modules -- a pre-transform and a post-transform -- represented by neural networks transforming the outputs of RO to a latent additive space (LAS), and then the LAS to action probabilities for execution; iii) two average oracles -- local average oracle (LAO) and global average oracle (GAO) -- where LAO operates on LAS and GAO is used for evaluation only; and iv) a novel method inspired by fictitious play that optimizes the transform modules and average oracles, and automatically selects the optimal combination of components of the two frameworks. Experiments on Leduc poker game demonstrate that our approach can outperform both frameworks.

\end{abstract}

\section{Introduction}
Extensive-form games provide a versatile framework capable of representing multiple agents, imperfect information, and stochastic events~\cite{shoham2008multiagent}. A lot of research has been done in the domain of two-player zero-sum extensive-games, and the developed methods have been successfully applied to real-world defender-attacker scenarios~\cite{tambe2011security}, heads-up limit/no-limit hold’em poker~\cite{bowling2015heads,brown2018superhuman,moravvcik2017deepstack}, multiplayer hold’em poker~\cite{brown2019superhuman}, and no-press diplomacy~\cite{gray2020human,bakhtin2021no}. 
% such as computing Nash equilibria by linear programs~\cite{shoham2008multiagent}, double oracle algorithms~\cite{mcmahan2003planning}, and counterfactual regret (CFR) minimization~\cite{zinkevich2007regret}. Meanwhile, these scalable algorithms have achieved many accomplishments. Double oracle algorithms ~\cite{mcmahan2003planning}, ~\cite{jain2011double}. Heads-up limit hold’em poker was essentially solved~\cite{bowling2015heads}. Later, significant progress has been made for heads-up no-limit hold’em poker based on CFR, such as Libratus ~\cite{brown2018superhuman} and DeepStack ~\cite{moravvcik2017deepstack}. 
Both in theory and applications, two paradigms proved themselves superior to their alternatives in extensive-form game solving: policy space response oracles (PSRO)~\cite{lanctot2017unified,muller2019generalized} which provide a unified solution of fictitious play~\cite{heinrich2015fictitious} and double oracle~\cite{mcmahan2003planning,jain2011double} algorithms, and counterfactual regret minimization (CFR)~\cite{zinkevich2007regret,lanctot2009monte}. Despite their initial success, scaling both approaches to games of real-world sizes proved troublesome, as action spaces reached hundreds of thousands of actions. This motivated the use of universal approximators capable of efficient representation of computationally demanding functions required in both CFR and PSRO. The resulting algorithms are commonly referred to as \enquote{deep} alternatives of thereof because they employ deep neural networks as approximators~\cite{brown2019deep,srinivasan2018actor,moravvcik2017deepstack,steinberger2020dream}. Together they form the framework of \textbf{deep equilibrium finding} (DEF), which we define as: 
\begin{displayquote}
\emph{Finding the equilibrium of games through methods empowered by deep neural networks.}
\end{displayquote}
% Based on this definition, we call NFSP~\cite{heinrich2016deep}, PSRO~\cite{lanctot2017unified} and Deep CFR~\cite{brown2019deep} as DEF algorithms.

% Because of their seemingly irreconcilable differences, both frameworks were often regarded as distinct.

Even though PSRO and CFR share some similarities, there are some key differences: i) PSRO methods use Q value to obtain the policy, while CFR methods rely on counterfactual regret values; ii) PSRO use meta-strategies, e.g., a Nash strategy, to average the strategies over action probabilities, while CFR methods take the average strategy to approximate Nash equilibrium; and iii) the sampling methods for computing new responses differ as well -- PSRO methods take the new response itself to sample in the game and CFR methods use the previous responses. These differences make both frameworks often regarded as distinct and the research on each framework is to a large extent evolving independently.
A very natural and important question often asked in the literature is then which framework is superior to the other. There are no clear answers as we observe contradicting evidence -- fictitious play is claimed to be superior to CFR in~\cite{ganzfried2020fictitious} while the opposite conclusion is reached in~\cite{brown2019deep}. We believe no simple answer exists as the performance of both frameworks seems intrinsically linked to structural properties of games, and their relation to efficient computability remain unexplored. Therefore, we analyze an alternative question: \emph{can we represent the two frameworks in a unified way and select among them automatically?}

To this end, we propose a unified framework for deep equilibrium finding. Our Unified DEF (UDEF) relies on following four main contributions: i) a novel response oracle (RO) which can compute Q values, reaching probability values and baseline values, which provides necessary values for both PSRO and CFR algorithms; ii) two transform modules, i.e., pre-transform and post-transform, represented by neural networks, first transform the outputs of RO to a latent additive space (LAS), which can be the action probabilities as PSRO or counterfactual values as CFR, and then transform from LAS to action probabilities for execution; iii) two average oracles, i.e., local average oracle (LAO) and global average oracle (GAO), where LAO operates on LAS and GAO is used for evaluation only; and iv) a novel method inspired by fictitious play to optimize the transform modules and average oracles, which can automatically select the optimal combinations of components of the two frameworks. Experiments on Leduc poker game demonstrate that UDEF can outperform both original frameworks. We believe that this work unify and generalize the two largely independent frameworks and can accelerate the research on both frameworks from a unified perspective.

% In this work, we aim to provide a unified framework for deep equilibrium finding generalizing the known DEF algorithms: PSROs~\cite{lanctot2017unified} (most importantly the neural fictitious self-play (NFSP)~\cite{heinrich2016deep}) and deep CFR~\cite{brown2019deep}. Our belief is that by altering and combining multiple DEF algorithms we may surpass the performance of many modern approaches. 

\section{Related Works}
% In this section, we present an overview of related works. 

The first line of works we base our approach on is derived from policy space response oracle (PSRO)~\cite{lanctot2017unified,muller2019generalized}, which is a generalization of the classic double oracle methods~\cite{mcmahan2003planning}. Given a game (e.g., poker), PSRO constructs a higher-level meta-game by simulating outcomes for all match-ups of a population of players’ policies. Then, it trains new policies for each player (via an oracle) against a distribution over the existing meta-game policies (typically an approximate Nash equilibrium, obtained via a meta-solver), appends these new policies to the meta-game population, and iterates. The widely-used neural fictitious self-play (NFSP)~\cite{heinrich2016deep} algorithm, which is an extension of fictitious play~\cite{brown1951iterative}, is a special case of PSRO. Recent works focus on improving the scalability~\cite{mcaleer2020pipeline,smith2020iterative}, improving the diversity of computed responses~\cite{perez2021modelling}, introducing novel meta-solvers, i.e., $\alpha$-rank and correlated equilibrium ~\cite{muller2019generalized,marris2021multi}, and applying to mean-field games~\cite{muller2021learning}. 

% There are three variants proposed in~\cite{srinivasan2018actor}, which can be viewed as linear transformation of the Q values plus the threshold operation.

% NFSP~\cite{heinrich2016deep}, PSRO~\cite{lanctot2017unified,muller2019generalized} 

The second line is counterfactual regret (CFR) minimization, which is a family of iterative algorithms that are the most popular approach to approximately solve large imperfect-information games~\cite{zinkevich2007regret}. The original CFR algorithm relies on the traverse of the entire game tree, which is impractical for large-scale games. Therefore, sampling-based CFR variants ~\cite{lanctot2009monte,gibson2012generalized} are proposed to solve large-scale games, which allows CFR to update regrets on parts of the tree for a single agent. Nowadays, neural network function approximation is applied to CFR to solve larger games. Deep CFR~\cite{brown2019deep} is a fully parameterized variant of CFR that requires no tabular sampling. Single Deep CFR~\cite{steinberger2019single} removes the need for an average network in Deep CFR and thereby enabled better convergence and more efficient training. 
Other extensions also demonstrate the effectiveness of neural networks in CFR~\cite{li2019double,li2021cfr,schmid2021player}.
% Deep CFR~\cite{brown2019deep}, Single Deep CFR~\cite{steinberger2019single}, Double Deep CFR~\cite{li2019double}.

\section{Preliminaries}
We focus on fully rational sequential interactions represented using imperfect-information extensive-form games. Imperfect-information games is a general model capable of modelling finite multiplayer games with limited observations. This section first describes this model. Later on, we explain how the two renowned DEF algorithms leverage it to compute rational strategies.

\subsection{Imperfect-Information Games}
An imperfect-information game (IIG) ~\cite{shoham2008multiagent} is a tuple ($N, H, A, P, \mathcal{I}, u$), where $N = \{1,...,n\}$ is a set of players and $H$ is a set of histories (i.e., the possible action sequences). The empty sequence $\emptyset$ corresponds to a unique root node of game tree included in $H$, and every prefix of a sequence in $H$ is also in $H$. $Z \subset H$ is the set of the terminal histories. $A(h) = \{a:(h,a) \in H\}$ is the set of available actions at a non-terminal history $h \in H$. $P$ is the player function. $P(h)$ is the player who takes an action at the history $h$, i.e., $P(h) \mapsto P \cup \{c\}$. $c$ denotes the ``chance player'', which represents stochastic events outside of the players' control. If $P(h) = c$ then chance determines the action taken at history $h$. Information sets $\mathcal{I}_i$ form a partition over histories $h$ where player $i \in \mathcal{N}$ takes action. Therefore, every information set $I_i \in \mathcal{I}_i$ corresponds to one decision point of player $i$ which means that $P(h_1) = P(h_2)$ and $A(h_1) = A(h_2)$ for any $h_1, h_2 \in I_i$. For convenience, we use $A(I_i)$ to represent the set $A(h)$ and $P(I_i)$ to represent the player $P(h)$ for any $h \in I_i$. For $i \in N$, $u_i: Z \rightarrow \mathbb{R}$ specifies the payoff of player $i$.

A player's behavioral strategy $\sigma_i$ is a function mapping every information set of player $i$ to a probability distribution over $A(I_{i})$. A joint strategy profile $\sigma$ consists of a strategy for each player $\sigma_1, \sigma_2, ...$, with $\sigma_{-i}$ referring to all the strategies in $\sigma$ except $\sigma_i$. Let $\pi^{\sigma}(h)$ be the reaching probability of history $h$ if players choose actions according to $\sigma$. Given a strategy profile $\sigma$, the overall value to player $i$ is the expected payoff of the resulting terminal node, $u_i(\sigma) = \sum_{h \in Z}\pi^\sigma(h)u_i(h)$. We denote all possible strategies for player $i$ as $\Sigma_{i}$. 

The canonical solution concept is Nash Equilibrium (NE). The strategy profile $\sigma$ forms a NE between the players if \begin{equation}
    u_{i}(\sigma) \geq u_{i}(\sigma_{i}', \sigma_{-i}), \forall i\in N.
\end{equation}
To measure of the distance between $\sigma$ and NE, we define $\textsc{NashConv}_{i}(\sigma)=\max_{\sigma_{i}'} u_{i}(\sigma_{i}', \sigma_{-i}) - u_{i}(\sigma)$ for each player and $\textsc{NashConv}(\sigma)=\sum_{i\in N}\textsc{NashConv}_{i}(\sigma)$. Other solution concepts such as (coarse) correlated equilibrium require different measures~\cite{marris2021multi}. Though our method is general and may be applied to different solution concepts, we focus on NE in this paper. 

\subsection{PSRO and CFR}
We present a brief introduction to existing DEF algorithms. 
\paragraph{PSRO.} PSRO is initialized with a set of randomly-generated policies $\hat{\Sigma}_{i}$ for each player $i$. At each iteration of PSRO, a meta-game $M$ is built with all existing policies of players and then a meta-solver computes a meta-strategy, i.e., distribution, over polices of each player (i.e., Nash, $\alpha$-rank or uniform distributions). The joint meta-strategy for all players is denoted as $\bm{\alpha}$, where $\alpha_{i}(\sigma)$ is the probability that player $i$ takes $\sigma$ as his strategy. After that, an oracle computes at least one policy for each player, which is added into $\hat{\Sigma}_{i}$. We note when computing the new policy for one player, all other players' policies and the meta-strategy is fixed, which corresponds to a single-player optimization problem and can be solved by DQN or policy gradient algorithms. NFSP can be viewed as a special case of PSRO with the uniform distributions as the meta-strategies. 

\paragraph{CFR.} Let $\sigma^t_i$ be the strategy used by player $i$ in round $t$. We define $u_i(\sigma, h)$ as the expected utility of player $i$ given that the history $h$ is reached and then all players play according to strategy $\sigma$ from that point on. Let us define $u_i(\sigma, h\cdot a)$ as the expected utility of player $i$ given that the history $h$ is reached and then all players play according to strategy $\sigma$ except player $i$ who selects action $a$ in history $h$. Formally, $u_i(\sigma,h)=\sum_{z \in Z}\pi^\sigma(h,z)u_i(z)$ and $u_i(\sigma, h\cdot a)=\sum_{z \in Z}\pi^\sigma(h \cdot a,z)u_i(z)$.
The \textit{counterfactual value} $u_i^\sigma(I)$ is the expected value of an information set $I$ given that player $i$ attempts to reach it. This value is the weighted average of the value of each history in an information set. The weight is proportional to the contribution of all players other than $i$ to reach each history. Thus, $u_i^\sigma(I) = \sum_{h\in I} \pi^\sigma_{-i}(h)\sum_{z \in Z}\pi^\sigma(h,z)u_i(z)$.
% \begin{align}
% u_i^\sigma(I) = \sum_{h\in I} \pi^\sigma_{-i}(h)\sum_{z \in Z}\pi^\sigma(h,z)u_i(z).  \nonumber
% \sum_{h\in I} \pi^\sigma_{-i}(h)u_i(\sigma,h) = \sum_{h\in I} \pi^\sigma_{-i}(h)\sum_{z \in Z}\pi^\sigma(h,z)u_i(z).  \nonumber
% \end{align}
For any action $a \in A(I)$, the counterfactual value of an action $a$ is $u_i^\sigma(I,a)= \sum_{h\in I} \pi^\sigma_{-i}(h)\sum_{z \in Z}\pi^\sigma(h \cdot a,z)u_i(z)$.
The \textit{instantaneous regret} for action $a$ in information set $I$ in iteration $t$ is $r^t(I,a)=u_{P(I)}^{\sigma^t}(I,a)-u_{P(I)}^{\sigma^t}(I)$. The cumulative regret for action $a$ in $I$ in iteration $T$ is $R^T(I,a) =\sum_{t=1}^{T}r^t(I,a)$. 
% Let $R^T_{+}(I,a)=\max(R^T(I,a),0)$ be the positive portion of cumulative regret because we often mostly concern about cumulative regret when it is positive. Actually, the regret for player $i$ in the entire game is $R^T_i = \max_{\sigma^{'}_i \in \sum_i} \sum_{t \in T}(u_i(\sigma^{'}_i,\sigma^t_{-i})-u_i(\sigma^{t}_i,\sigma^t_{-i}))$. 
In CFR, players use \textit{Regret Matching} to pick a distribution over actions in an information set proportional to the positive cumulative regret on those actions. Formally, in iteration $T+1$, player $i$ selects actions $a\in A(I)$ according to probabilities
% \normalsize
% \begin{small}
\begin{equation}
\sigma^{T+1}(I,a)=
\begin{cases}
\frac{R^T_{+}(I,a)}{\sum_{b \in A(I)}R^T_{+}(I,b)} & \text{if $\sum\limits_{b \in A(I)}R^T_{+}(I,b) >0$} \\
\frac{1}{|A(I)|}& \text{otherwise,} 
\end{cases}\nonumber
\end{equation}
% \end{small}\normalsize
where $R^T_{+}(I,a)=\max\{R^T(I,a),0\}$ because we are concerned about the cumulative regret when it is positive only.
% The regret of information set in iteration $T$ is $R^T(I) = \max_a(R^T_{+}(I,a))$. 
If a player acts according to CFR in every iteration, then in iteration $T$, $R^T(I) \leq \Delta_i\sqrt{|A_i|}\sqrt{T}$ where $\Delta_i= \max_z u_i(z)-\min_zu_i(z)$ is the range of utility of player $i$. Moreover, 
$R^T_i \leq \sum_{I\in \mathcal{I}_i}R^T(I) \leq |\mathcal{I}_i|\Delta_i\sqrt{|A_i|}\sqrt{T}$. 
Therefore, $\lim_{T \rightarrow \infty}\frac{R^T_i}{T} = 0$. In two-player zero-sum games, if both players' average regret  $\frac{R^T_i}{T} \leq \epsilon$, their average strategies $(\overline{\sigma}^T_1,\overline{\sigma}^T_2)$ form a $2\epsilon$-equilibrium~\cite{waugh2009abstraction}. Most previous works focus on tabular CFR, where counterfactual values are stored in a table. Recent works adopt deep neural networks to approximate the counterfactual values and outperform their tabular counterparts~\cite{brown2019deep,steinberger2019single,li2019double,li2021cfr}. 
\section{Unified Deep Equilibrium Finding (UDEF)}
\begin{algorithm}[t]
\caption{Deep Equilibrium Finding (DEF)}
\label{alg:def}
\begin{algorithmic}[1]
\STATE policy lists for players $\mathcal{L}=\times\mathcal{L}_{k}$
\WHILE{TO is \textbf{False}}
\STATE Local AO (LAO) to obtain average policies $\bar{\mathcal{L}}=\times\bar{\mathcal{L}_{k}}$, which is used for the next iteration
\STATE RO to find new responses against opponents' policies
\ENDWHILE
\STATE Global AO (GLO) for evaluation.
\end{algorithmic}
\end{algorithm}
After introducing the context, we move to the main part of the paper where we present the unified perspective of the deep equilibrium finding. The existing DEF algorithms may be represented as Algorithm~\ref{alg:def}, which consists three oracles:
\begin{itemize}
\item \textbf{Response Oracle} (RO): The response oracle identifies an ``optimal'' response to the opponents' static policies. The response will specify the behavioral strategy at each infoset, based on either Q values as in PSRO or the counterfactual values as in CFR. %How you respond to your opponents when their policies are fixed.
\item \textbf{Average Oracle} (AO): There are two AOs. The local AO (LAO) is used to compute the sampling policy to sample the experiences from the game, and the global AO (GAO) is used for evaluation. Specifically, the LAO in PSRO can be Nash equilibrium strategies or the $\alpha$-rank distribution and the GAO is the last strategy used for sampling. In CFR, the LAO can be uniform or linear average, as well as the GAO. 
\item \textbf{Termination Oracle} (TO) determines whether to end the learning process or continue, which is most often driven by the chosen solution concept, e.g., the exploitability in case of the Nash equilibrium. %How you determine whether or not terminate the learning process.
\end{itemize}
The structure of the unified DEF (UDEF) framework proposed in this work is displayed in Figure~\ref{fig:ro_with_tran}. After
RO
which computes the new response (denoted as new resp.) we plug the pre-transform module (denoted as pre-tran.)
to interpret the results to an linear additive space (LAS), where LAO
can take the average of the responses -- the historical responses 
(denoted as his. resp.) and the new response (new resp.) -- in the hidden space. Then, we plug the post-transform module (denoted as post-tran.) to transform the results from the linear additive space to the policy which is used in the next iteration. The GAO is also used for evaluation only. 

\begin{figure}[ht]
\centering
\includegraphics[width=0.35\textwidth]{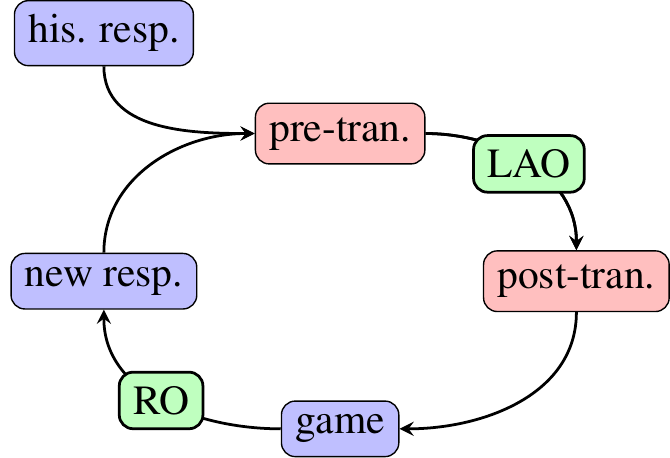}
\caption{Unified Deep Equilibrium Finding (UDEF)}
\label{fig:ro_with_tran}
\end{figure}

% \begin{figure}[ht]
% \centering
% \subfigure[RO]{
% \includegraphics[width=0.175\textwidth]{figures/ro.pdf}
% \label{fig:ro}
% }
% \subfigure[LAO]{
% \includegraphics[width=0.175\textwidth]{figures/ao.pdf}
% \label{fig:ao}
% }
% \subfigure[Pre-tran.]{
% \includegraphics[width=0.175\textwidth]{figures/pre.pdf}
% \label{fig:pre}
% }
% \subfigure[Post-tran.]{
% \includegraphics[width=0.175\textwidth]{figures/post.pdf}
% \label{fig:post}
% }
% \caption{Design of neural networks in UDEF}
% \label{fig:transform}
% \end{figure}

\subsection{Design of Oracles and Modules in UDEF}
Here we describe the design of the neural networks in the oracles and modules in UDEF, which plays a fundamental role in generalizability of UDEF. 
\paragraph{Response Oracle.}
The response oracle computes the player's new policy against the average policy of the opponents. In PSRO and its variants, the response oracle returns a Q-value network, where the action with the maximum Q value is taken at each infoset. In CFR and its variants, the response oracle returns the counterfactual regret values, and the action is picked according to regret matching. We observe that the counterfactual regret values may be computed from the Q values, with a baseline value (BV) computed through the sampling policy. For MC-CFR~\cite{Lanctot09mccfr}, the reaching probability is also important for computing the counterfactual values. Any attempt for unification thus relies on estimating both the Q values and the reaching probabilities (RP). Due to estimation errors, when the infoset is closer to the root node, the reaching probability estimation is more accurate; when the infoset is closer to the terminal nodes, the Q value estimation is more accurate. With both Q value and reaching probability, we may hence balance the usage of both for more accurate decision making. For computing the counterfactual regret values from Q values, we introduce a policy network to track the sampling policy and use it to compute the baseline values efficiently. For their usage in RO, the Q values and RP values are obtained from the Q and RP networks, respectively, and the BV is calculated by taking the average of Q values and the action probabilities output by the policy network. To summarize, the RO includes three networks: i) a Q value network to estimate the Q value\footnote{A practical implementation of DQN~\cite{mnih2015human} employs an online Q network and a target Q network for TD-learning.}, ii) an RP network to estimate the reaching probability of the current infoset, and iii) a policy network which tracks the sampling policy and is used for computing counterfactual regret values. 
\paragraph{Average Oracles.}
% For the average oracle, we will use LSTM or transformer to generate a vector which associates a (positive) real value to each response. Then we use the vector to average the responses to obtain the averaged response. 

We design two average oracles. For the local average oracle, we use the network introduced in~\cite{feng2021discovering} and generalize it to represent the linear average function used in CFR variants. For GAO, we can also use the same structure as LAO to make the oracle differentiable, or use explicit functions as adopted in PSRO and CFR. 

\paragraph{Transform Modules.}

The two transform modules play important roles in generalizability of the UDEF framework. The pre-transform module takes the Q values, the reaching probabilities (RP) and the baseline values (BV), which is the weighted average of the Q values with the weights outputed by the policy network in RO, as the input and has two branches as the outputs. For the first branch, we use Parametric Rectified Linear Unit (\texttt{PReLU})~\cite{he2015delving} as the activation function for each dimension of LAS where
\begin{align}
\text{PReLU}(y_{i})=
\begin{cases}
y_{i}, & \text{if} \quad y_{i}\geq 0 \\
a_{i}\cdot y_{i}, & \text{if} \quad y_{i}< 0,
\end{cases}
\end{align}
where $a_{i}$ is optimized. We note that if $a_{i}=0$, PReLU becomes similar to CFR$^{+}$ and if $a_{i}=1$, it is equal to vanilla CFR. The same approach is adopted in~\cite{srinivasan2018actor}. The second branch uses \texttt{softmax} as the activation function. The post-transform module transforms the values in LAS to the probability distribution over legal actions. Therefore, we use \texttt{softmax} function as the activation function of the final layer. To handle the illegal actions, we include a large negative value as the penalty to the values of the illegal actions before the \texttt{softmax} operation of both pre-transform and post-transform modules, which ensures the probability of selecting illegal actions remains zero. 

% The designs of the oracles and modules allows UDEF to represent both PSRO and CFR algorithms and even generalize to more general algorithms for DEF. 

\subsection{Optimization of UDEF}
Given the architecture of different modules in UDEF, in this section, we provide the details of the training procedure. We first introduce the pretraining scheme of the AOs and transform module to speed up the training and then describe how the modules are optimized.

\textbf{Pretraining.} The random initialization of the AOs and the transform modules may result in the training being difficult. We hence pretrain them with either PSRO or CFR as the staring point. The two pretraining scheme are:
\begin{itemize}
\item \emph{CFR-pretraining.} We first pre-train the AO to a linear average function. We also pretrain the pre-transform module to compute the counterfactual regret values and the post-transform module to do the regret-matching. 
\item \emph{PSRO-pretraining.} We first pre-train the AO to the following solution concepts: Nash strategy, uniform strategy, and $\alpha$-rank strategy. The pre-transform module and the post-transform module are pre-trained to be softmax and identity function, respectively. 
\end{itemize}
% The pretrained AOs and transform modules will be used as the starting points for the optimization.
\begin{algorithm}
\caption{Optimization of RO}
\label{alg:dqn_rp}
\begin{algorithmic}[1]
\STATE Initialize the Q and target Q network $\theta$ and $\bar{\theta}$
\STATE Initialize the RP network $\phi$ and the policy network $\zeta$
\WHILE{Not terminated}
\STATE Simulate the game with sampling policies of players, and push $\langle s, a, s', r, rp_{s}, rp_{s'}\rangle$ to replay buffer\label{alg:dqn_rp_rollout}
\STATE Distill sampling policy into a new policy network $\zeta'$~\label{alg:dqn_rp_distill}
\STATE Train the Q network by minimizing the TD-error:
$
\theta = \theta - \nabla [r+\mathbb{E}[\bar{\theta}(s')]-\theta(s)]^{2}
$~\label{alg:dqn_rp_td}
\STATE Update the target Q network:  $\bar{\theta}=\tau \bar{\theta}+(1-\tau)\theta$
\STATE Train the RP network by minimizing the MSE error: \\
$\phi = \phi - \lambda_{\phi} \nabla [(rp_{s}-\psi(s))^{2} + (rp_{s'}-\psi(s'))^{2}]
$~\label{alg:dqn_rp_rp}
\STATE Update $\zeta$ with $\zeta'$~\label{alg:dqn_rp_zeta}
\ENDWHILE
\end{algorithmic}
\end{algorithm}
\paragraph{Optimizing RO.} The algorithm for optimizing the RO is shown in Algorithm~\ref{alg:dqn_rp}. We first simulate the game by sampling policies of players to generate new experiences, both the transition of states and the reaching probability into the replay buffer (Line~\ref{alg:dqn_rp_rollout}). One of the key differences of PSRO and CFR lies in the sampling method for generating new experiences: PSRO uses the new response to sample the experiences, while CFR uses the historical responses. We hence introduce a smooth parameter $hs\in[0, 1]$ to balance the sampling with historical responses and the new response. If $hs=0$, we use the historical responses with the transform modules for generating the new experience, which corresponds to the sampling method of CFR. If $hs=1$, we use the new response with the transform modules. After generating the new experiences, we distill the sampling policy into a new policy network $\zeta'$~(Line~\ref{alg:dqn_rp_distill}), where the policy network is the sampling of the current iteration and is used for computing baseline values in the next iteration. We use reinforcement learning methods, e.g., DQN~\cite{mnih2015human}, to estimate the Q values through minimizing the TD-error where $\mathbb{E}[\bar{\theta}(s')]$ is the expected Q values of $s'$~(Line~\ref{alg:dqn_rp_td}).
We use the sampling policy to compute $\mathbb{E}[\bar{\theta}(s')]$, which differs from the $max$ action used in DQN~\cite{mnih2015human}. Therefore, we distill the policy into $\zeta'$ first and only update $\zeta$ after training both Q and RP networks~(Line~\ref{alg:dqn_rp_zeta}).
We use the supervised method to train the RP network. The RP network is then employed to predict the reaching probabilities of both $s$ and $s'$ for a transition~(Line~\ref{alg:dqn_rp_rp}).

\begin{algorithm}
\caption{Optimizing of AOs and transform modules}
\label{alg:meta_optim}
\begin{algorithmic}[1]
\STATE Given the current joint strategy profile $\sigma$, the two AOs and two transform modules' parameters $\psi=\{\psi_{lao}, \psi_{gao}, \psi_{pre}, \psi_{post}\}$
\FOR{$i\in N$}
\STATE Compute the best response $\sigma'_{i}$ against $\sigma_{-i}$
\ENDFOR
\STATE Train the parameters in $\psi$ to approximate the joint policy $\eta\sigma_{i}+(1-\eta)\sigma'_{i}$ with the loss of action probability
\end{algorithmic}
\end{algorithm}
\paragraph{Optimizing AOs and Transform Modules.} Our objective is to optimize the AO and the two transform modules to find the optimal parameters of UDEF~\cite{feng2021discovering}. However, it is difficult to compute the gradients of $\textsc{NashConv}(\sigma)$ in IIGs, especially for large domains. Therefore, we employ the idea of fictitious play~\cite{heinrich2015fictitious}. Suppose that the current policy for evaluation is $\sigma$, and $\sigma'_{i}$ is a new best response against $\sigma_{-i}$. We optimize the three modules to approximate the policy $\eta\sigma_{i}+(1-\eta)\sigma'_{i}$, where $\eta$ is a smoothing parameter. The details of the process are described in Algorithm~\ref{alg:meta_optim}, and proposition~\ref{prop:exploit} proves that it decreases $\textsc{NashConv}(\sigma)$. After the update of LAO, GAO and transform modules, we recompute the meta-game, as well as the values for LAO and GAO, and start a new iteration. We may also use the idea of PSRO, where the meta-game is built between the policy $\sigma$ and new responses $\sigma'_{i}, \forall i\in N$, and NE or $\alpha$-rank distribution~\cite{lanctot2017unified,muller2019generalized} is used to generate the target policy during the optimization. We will investigate this idea and other efficient optimization methods in future works.

\begin{proposition}
\label{prop:exploit}
The $\textsc{NashConv}(\sigma)$  is decreased through approximating the new average policy $\eta\sigma_{i}+(1-\eta)\sigma'_{i}$.
\end{proposition}
\begin{proof}
For each player $i$, suppose the utility of playing $\sigma_{i}$ against $\sigma_{-i}$ is $u_{i}$ and the utility of playing $\sigma'_{i}$ against $\sigma_{-i}$ is $u'_{i}$. Because $\sigma'_{i}$ is the best-response, $u'_{i}\geq u_{i}$. If $u'_{i}= u_{i}$, $\textsc{NashConv}_{i}(\sigma)=0$ for $i$ and the modules are not updated. In case $u'_{i}>u_{i}$, the utility of playing $\eta\sigma_{i}+(1-\eta)\sigma'_{i}$ is $\eta u_{i} + (1-\eta) u'_{i}$, which is strictly greater than $u_{i}$. Therefore, the $\textsc{NashConv}_{i}(\sigma)$ of player $i$ is decreased when training AO and transform modules to approximate the policy $\eta\sigma_{i}+(1-\eta)\sigma'_{i}$. This result holds for each player, $\textsc{NashConv}(\sigma)$ thus can be decreased. 
\end{proof}

\section{Analysis of UDEF}
In this section, we prove the equivalence of UDEF to PSRO and CFR under specific configurations. As such, UDEF unifies and generalizes the two main directions of DEF.

\begin{proposition}[UDEF $\rightarrow$ PSRO]
Let the pre-transform module be equal to softmax and the post-transform module to the identity function. Then UDEF reduces to PSRO.
\end{proposition}
\begin{proof}
As RO outputs the Q values, when the pre-transform module is the softmax, the output of the pre-transform module is a distribution over actions, which corresponds to the behavioral strategy. Note that the Q values of illegal actions are penalized with a large negative value, which ensures that the probabilities of illegal action are always zero. PSRO will operate on the action probabilities, i.e., the values in LAS correspond to the action probabilities. After applying the LAO on the outputs of all responses, where the values can correspond to a uniform distribution, an NE strategy or an $\alpha$-rank distribution, we obtain the averaged values in LAS, which is still a distribution over actions. If the post-transform module is an identity mapping, the output of this oracle is exactly the softmax distribution over Q values. Because PSRO makes computing the best-response a single-agent RL problem, we use only the Q values from the new response for sampling the new experiences, i.e., $hs=0$. PSRO uses the policy the opponent responses to for evaluation, therefore, the GAO is a vector with zero values everywhere except for the last element, which is one.
\end{proof}

\begin{proposition}[UDEF $\rightarrow$ CFR]
Let the pre-transform module map Q$\rightarrow$ Regret values and the post-transform module represent regret matching. Then UDEF reduces to CFR.
% With the Q to Regret values for pre-transform module and the regret matching for the post-transform module, UDEF is reduced to CFR.
\end{proposition}
\begin{proof}
The Q values, RP values and the BV are obtained from the RO. The training procedure displayed in Algorithm~\ref{alg:dqn_rp} ensures that the BV is the average value of the Q values weighted by the sampling policy. The counterfactual regret value is a linear combination of Q values, RP values, and the BV. As such, it can be approximated by the pre-transform module represented by a neural network. The LAO is uniform in vanilla CFR~\cite{zinkevich2007regret} and a linear average in linear CFR~\cite{brown2019solving}. Both of the LAOs can be approximated by a neural network as well. The post-transform module implements regret matching, which can also be approximated by a neural network. One of the key differences between PSRO and CFR lies in the fact that the policy CFR uses for evaluation differs from the policy the opponent plays against. The GAO in CFR is the uniform distribution in vanilla CFR and the linear average in the linear CFR, which can both be also represented by a neural network. Because CFR utilizes the historical responses to generate new experiences, we set $hs=1$ as the sampling policy. 
\end{proof}

\begin{table}[ht]
\caption{Summary of different algorithms.} 
\label{tab:summary}
\centering
% \begin{small}
\begin{tabular}{c|l|c|c|c|c}
\toprule\toprule
\multicolumn{2}{c|}{} & NFSP & PSRO & CFR & LCFR\\
\midrule
\multicolumn{2}{c|}{RO}& \multicolumn{2}{c|}{DQN} & \multicolumn{2}{c}{DQN+RP+BV}\\
\midrule
\multirow{2}{*}{AO}&LAO & Uni. & Nash& Uni. & Lin.\\
\cmidrule{2-6}
&GAO & \multicolumn{2}{c|}{$[0, \dots, 1]$} & Uni.&Lin.\\
\midrule
\multicolumn{2}{c|}{Pre-tran.} & \multicolumn{2}{c|}{Softmax} & \multicolumn{2}{c}{Regret}\\
\cmidrule{1-6}
\multicolumn{2}{c|}{Post-tran.} & \multicolumn{2}{c|}{Identity} & \multicolumn{2}{c}{Match} \\
\midrule
\multicolumn{2}{c|}{Sampling} & \multicolumn{2}{c|}{new response} & \multicolumn{2}{c}{historical responses} \\
%  & LAO & RO & Pre-tran. & Post-tran.\\\midrule\midrule
% NFSP & Uni. & DQN & Argmax & Idn.\\\midrule
% PSRO & Meta & DQN & Argmax & Idn.\\\midrule
% % PSRO & Meta & REINFORCE & Idn. & Idn.\\\midrule
% % XDO & Meta & CFR & Match & Idn.\\\midrule
% CFR & Lin. & \textcolor{red}{DQN+RP+BV} & \textcolor{red}{Regret} & Match\\
\bottomrule\bottomrule
\end{tabular}
% \end{small}
\end{table}

A summary of the correspondences of different algorithms with the configurations of UDEF is depicted in Table~\ref{tab:summary}. Whenever the average oracles and transform modules are represented by neural networks, we may efficiently shift between different algorithms in a differentiable way. UDEF hence operates on a much more general space than any of the existing algorithms, which can be potentially used to search for more powerful equilibrium finding algorithms. 

\section{Experiments}

\begin{figure*}[!htbp]
%\vspace{-9pt}
\centering
\subfigure[NFSP with $hs=0.1$]{\includegraphics[width=0.325\textwidth]{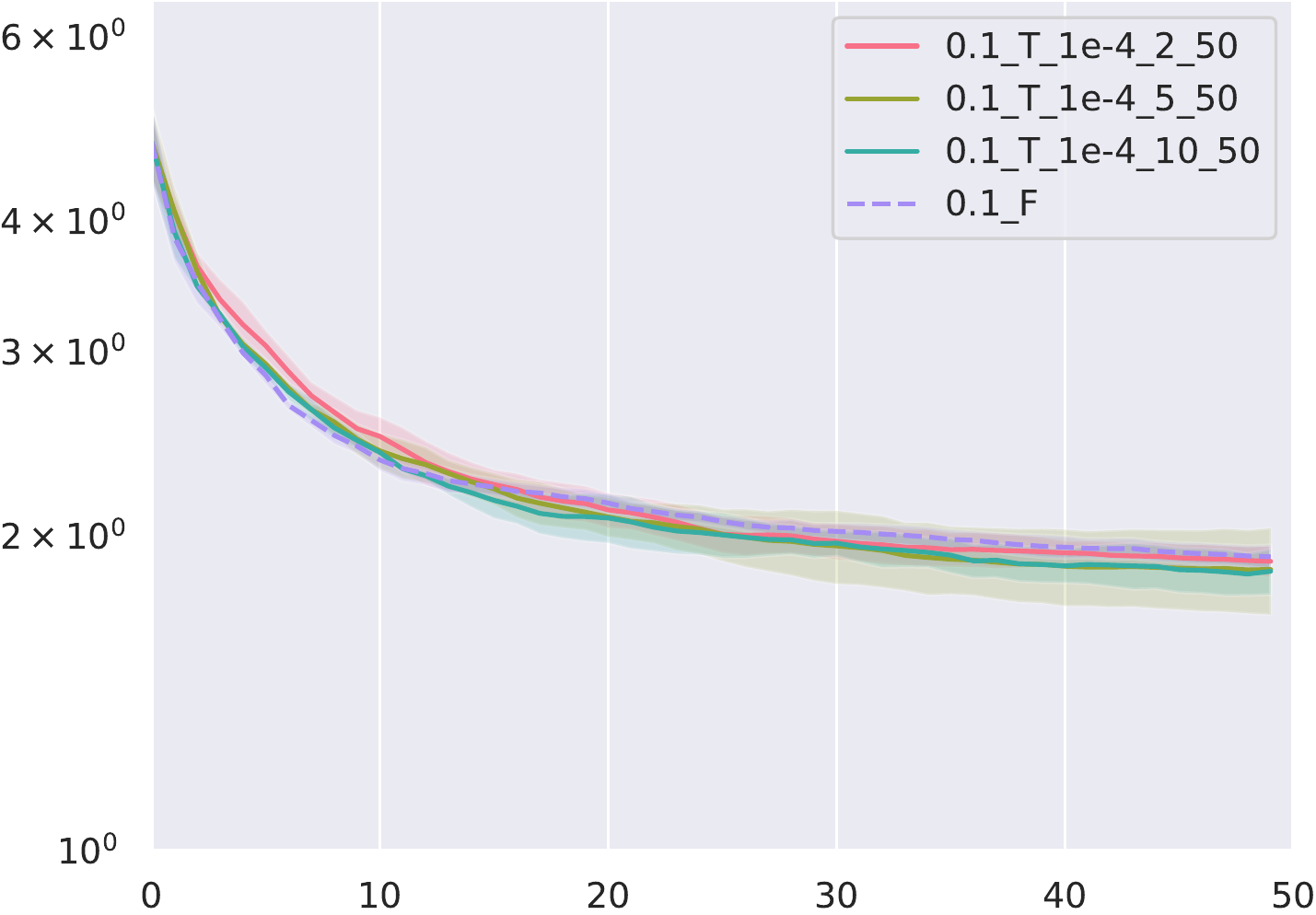}\label{fig:udef_nfsp_0_1}}
\subfigure[NFSP with $hs=0.5$]{\includegraphics[width=0.325\textwidth]{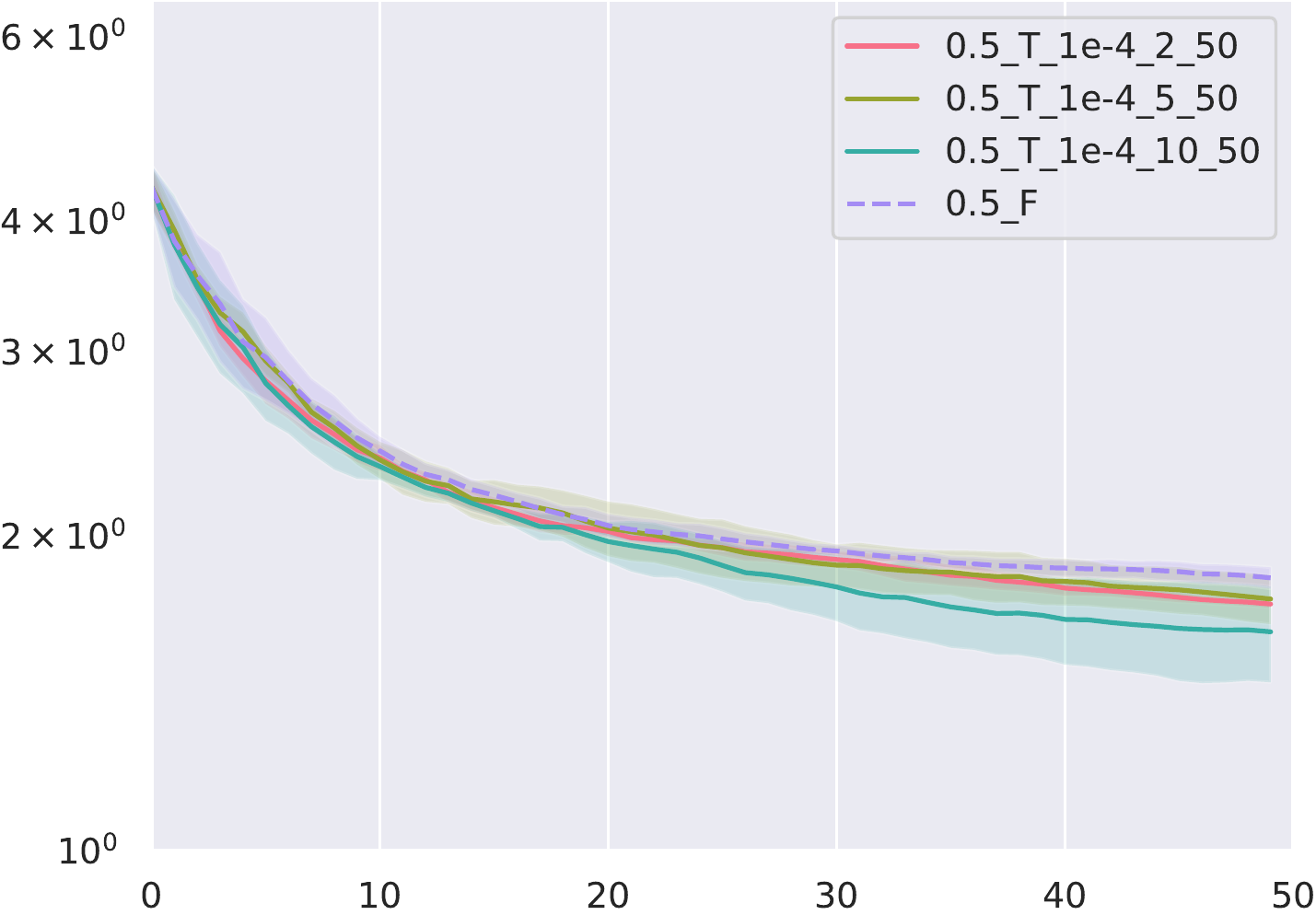}\label{fig:udef_nfsp_0_5}}
\subfigure[NFSP with $hs=0.9$]{\includegraphics[width=0.325\textwidth]{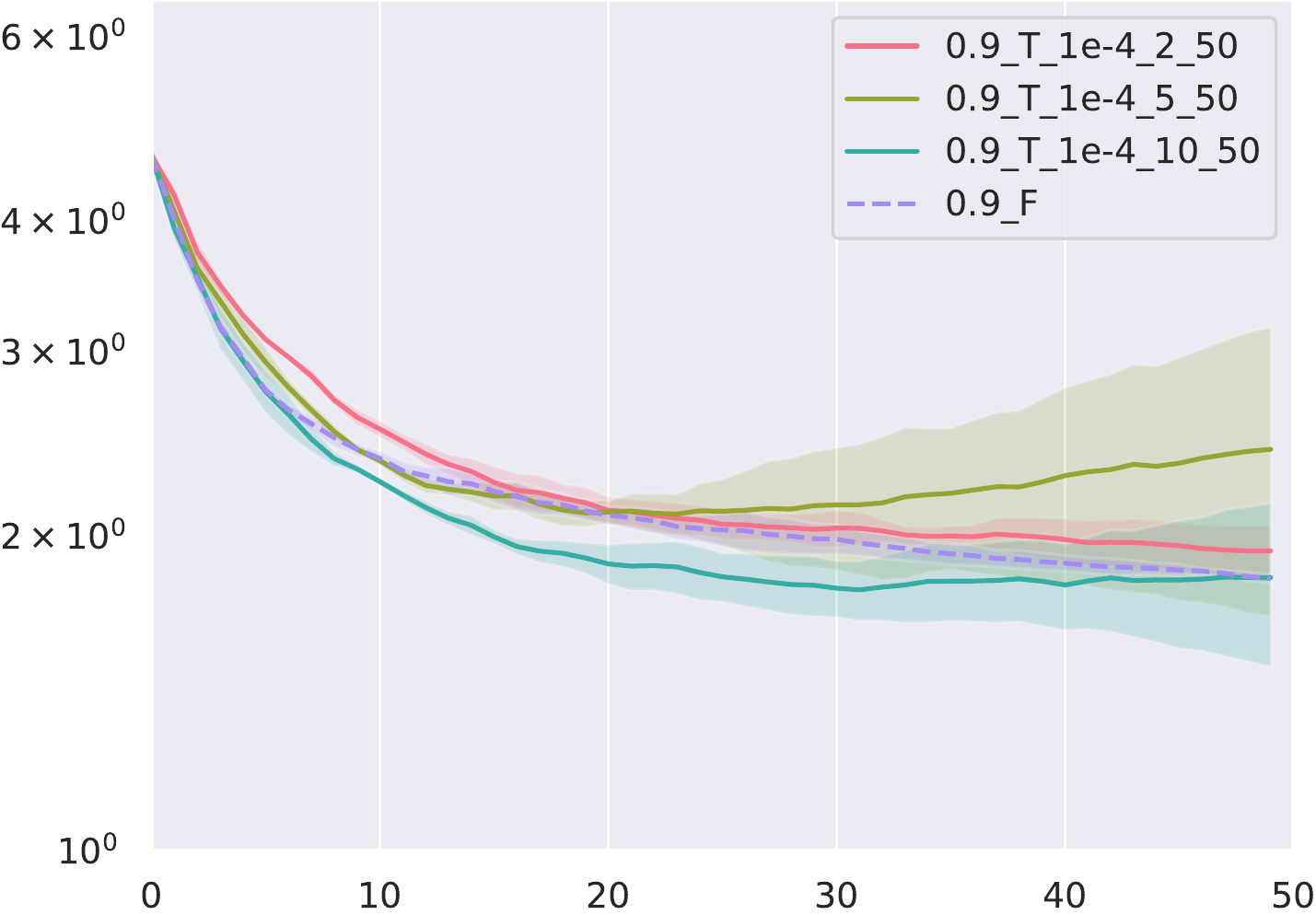}\label{fig:udef_nfsp_0_9}}
\subfigure[CFR with $hs=0.1$]{\includegraphics[width=0.325\textwidth]{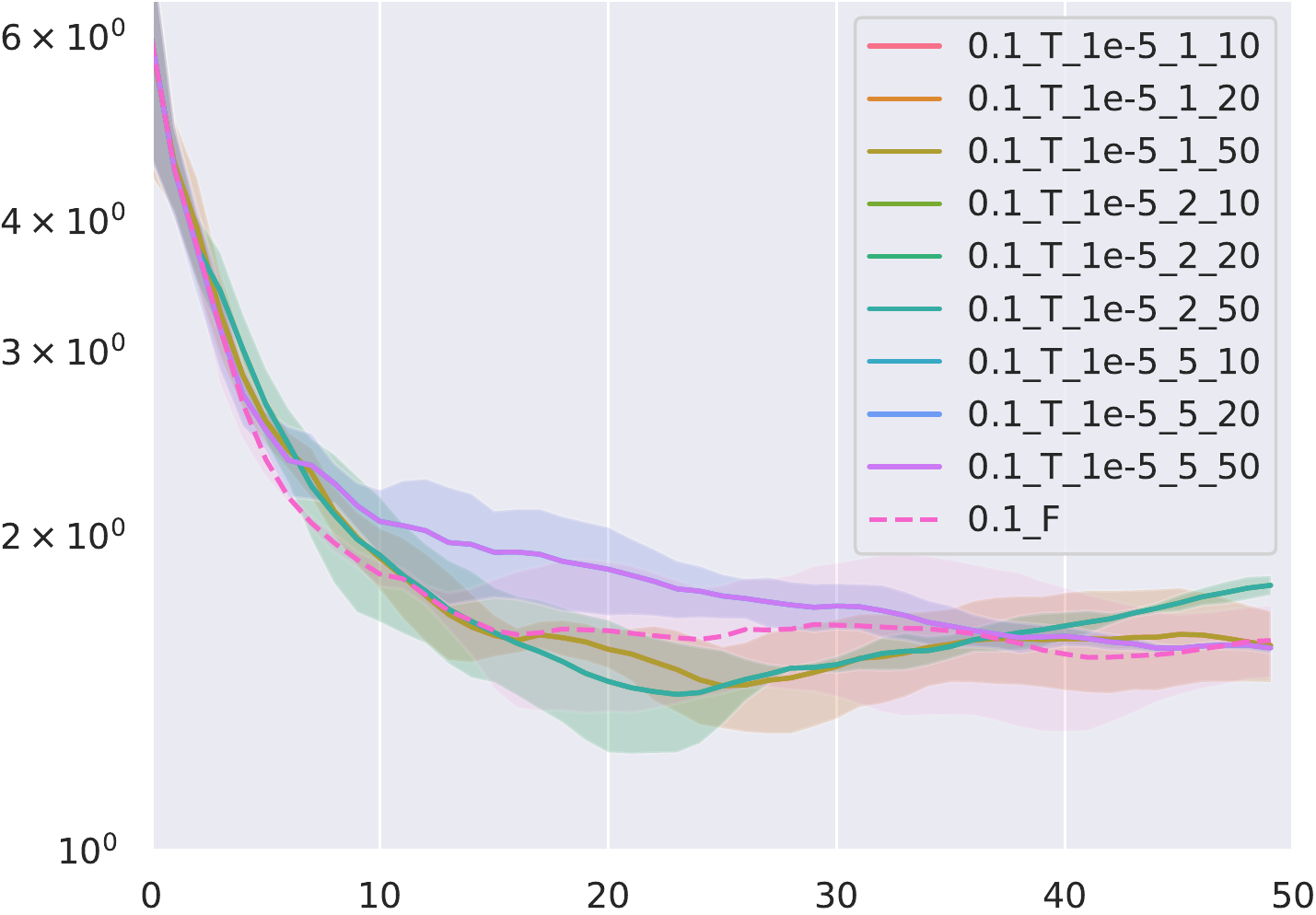}\label{fig:udef_cfr_0_1}}
\subfigure[CFR with $hs=0.5$]{\includegraphics[width=0.325\textwidth]{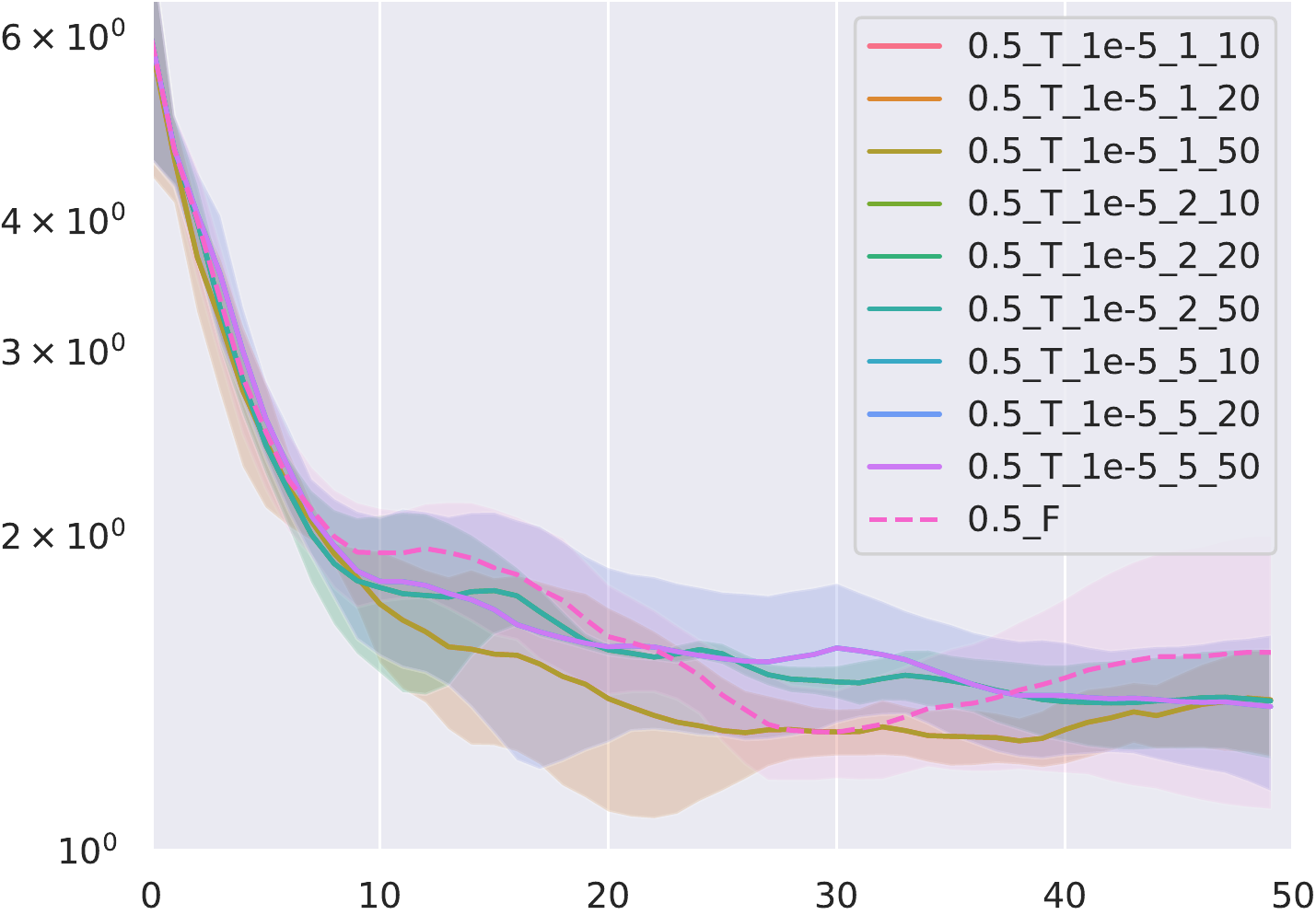}\label{fig:udef_cfr_0_5}}
\subfigure[CFR with $hs=0.9$]{\includegraphics[width=0.325\textwidth]{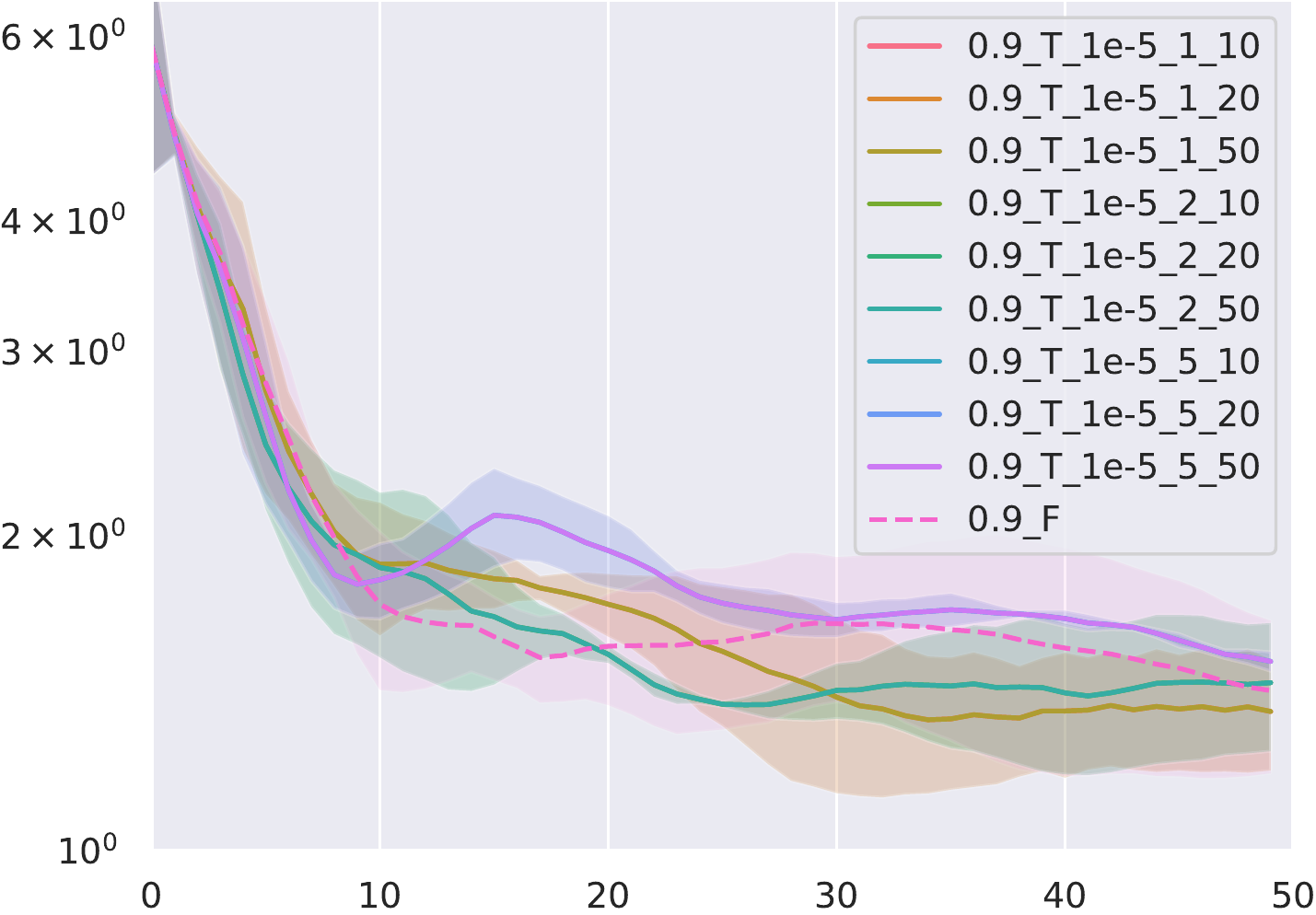}\label{fig:udef_cfr_0_9}}
%\vspace{-9pt}
\caption{Results of \textsc{NashConv} of UDEF with optimizing transform modules only in NFSP and Deep CFR. The legend in the figure should read as $\langle hs, meta\_train, meta\_lr, meta\_steps, meta\_train\_max\rangle$.}
\label{fig:tran_opt}
%\vspace{-9pt}
\end{figure*}

In this section, we present the experiment results of UDEF.
\subsection{Experiment Setup}
We focus on Leduc poker, which is a widely used experimental domain in equilibrium finding~\cite{southey2005bayes,lanctot2019openspiel,muller2019generalized}. We compare our UDEF algorithm to NFSP~\cite{heinrich2015fictitious} and Deep CFR~\cite{brown2019deep}. To this end, we pretrain the pre-transform and the post-transform modules with PSRO and CFR with a batch of 1000 for $10^4$ episodes. The DRL algorithm used in RL is DQN~\cite{mnih2015human}, where an online and a target Q networks are introduced to perform the TD-learning. The Q, RP and policy networks in RO are configured as multilayer perceptrons (MLPs) with one hidden layer containing 256, 64 and 256 hidden neurons, respectively. We use MSE loss and employ the Adam optimizer~\cite{kingma2015adam} with a default learning rate of $0.01$ for training the RO. The sampling policy generating new experiences uses $hs\in\{0.1, 0.5, 0.9\}$ to represent the approximation of PSRO and CFR and the combination of both. For the distillation of the policy, we train the policy for 5 episodes.  The two transform modules are MLPs with one hidden layer of 64 neurons. We choose the dimension of LAS as 16, which is more than twice the number of actions in Leduc poker, i.e., 3, and enables us to pretrain the transform modules with both PSRO and CFR methods simultaneously. Similarly to policy networks, we use MSE loss and Adam optimizer with a learning rate of $5\cdot 10^{-3}$ for pretraining. For the training of AO and transform modules, we choose $\eta=0.95$ to smooth the current average policy and the new response. Throughout the experiment, we vary the values of the learning rate, the update steps and the max number of iterations to optimize AO and transform modules. We denote these values as $meta\_lr$, $meta\_steps$ and $meta\_train\_max$, respectively. The values differ in different algorithms, as described in next sections. In each iteration, we sample the experiences from game with $10^4$ episodes for RO and $10^3$ episodes for augmenting the meta-game. The maximum number of iterations is 50. All results are averaged over 3 seeds. 

\begin{figure*}[ht]
%\vspace{-9pt}
\centering
\subfigure[NFSP with $hs=0.1$]{\includegraphics[width=0.325\textwidth]{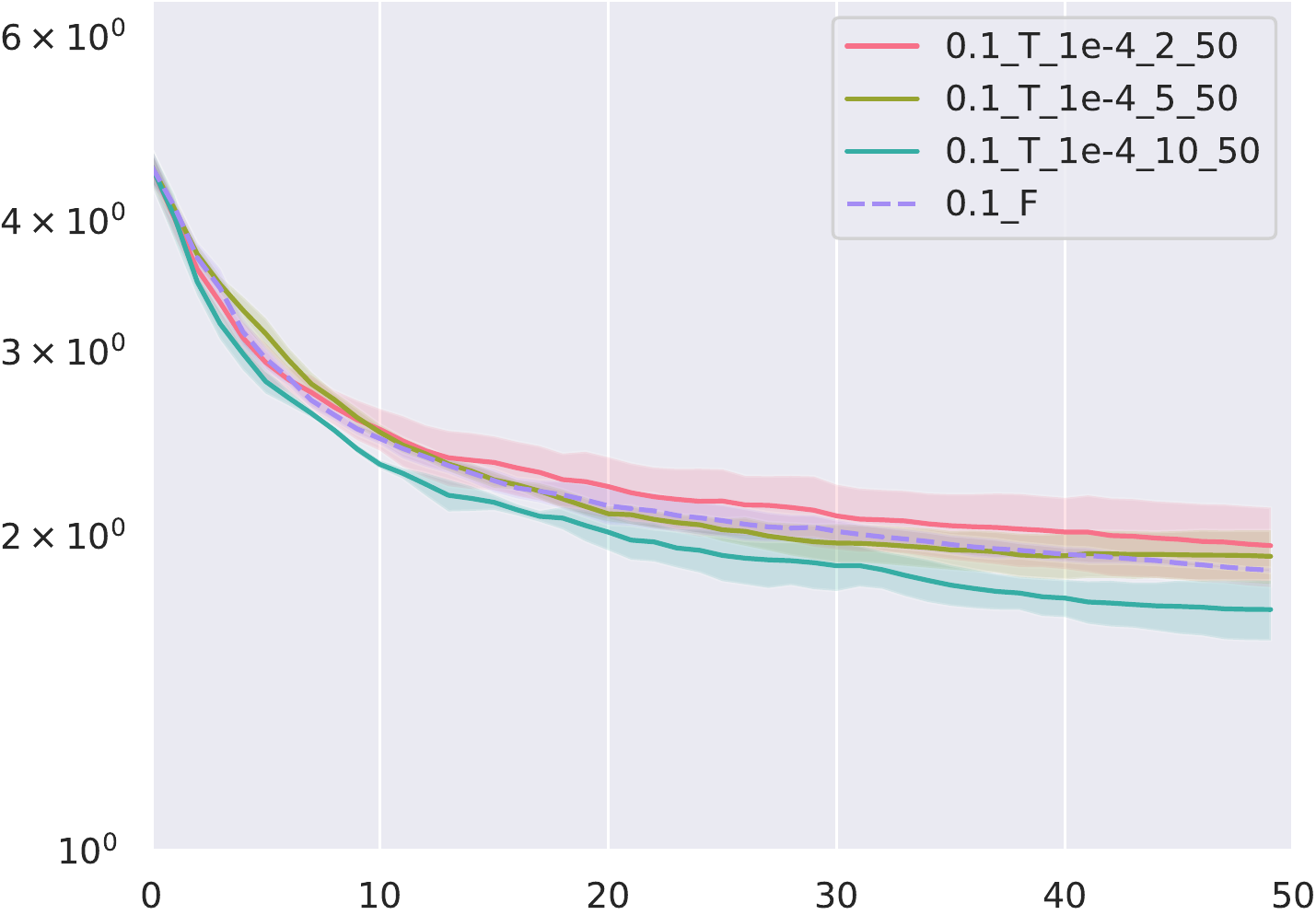}\label{fig:udef_neural_nfsp_0_1}}
\subfigure[NFSP with $hs=0.5$]{\includegraphics[width=0.325\textwidth]{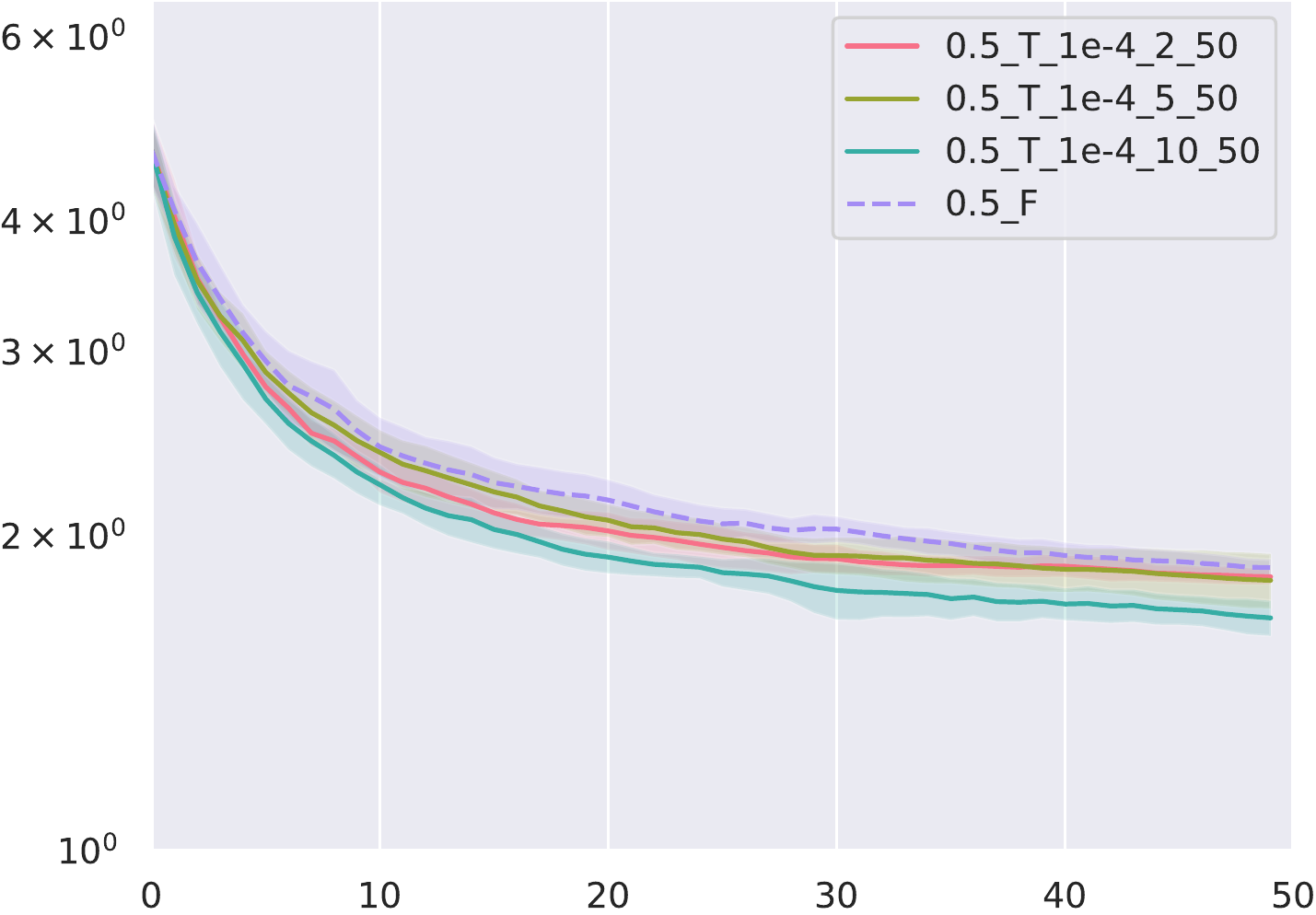}\label{fig:udef_neural_nfsp_0_5}}
\subfigure[NFSP with $hs=0.9$]{\includegraphics[width=0.325\textwidth]{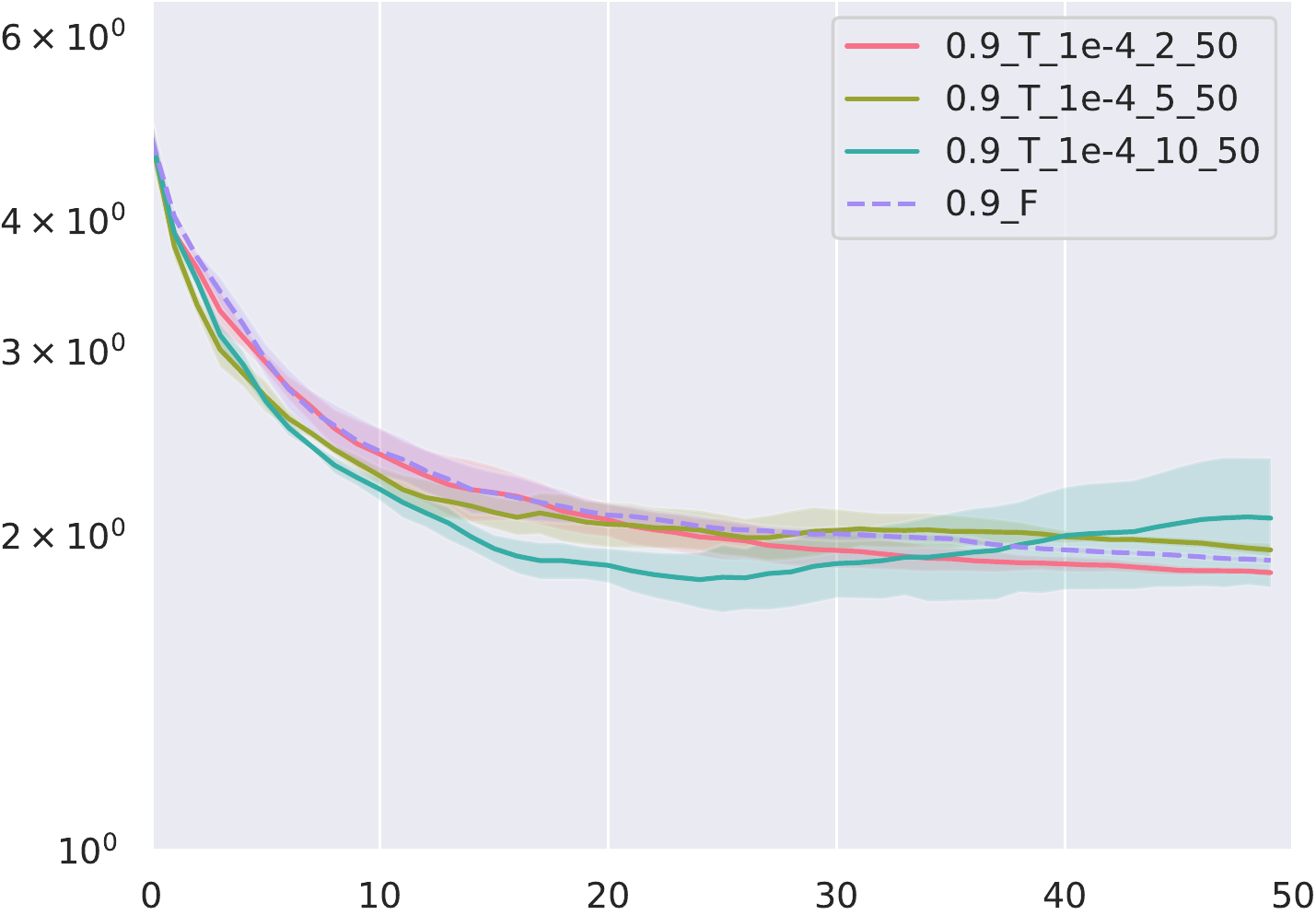}\label{fig:udef_neural_nfsp_0_9}}
\subfigure[CFR with $hs=0.1$]{\includegraphics[width=0.325\textwidth]{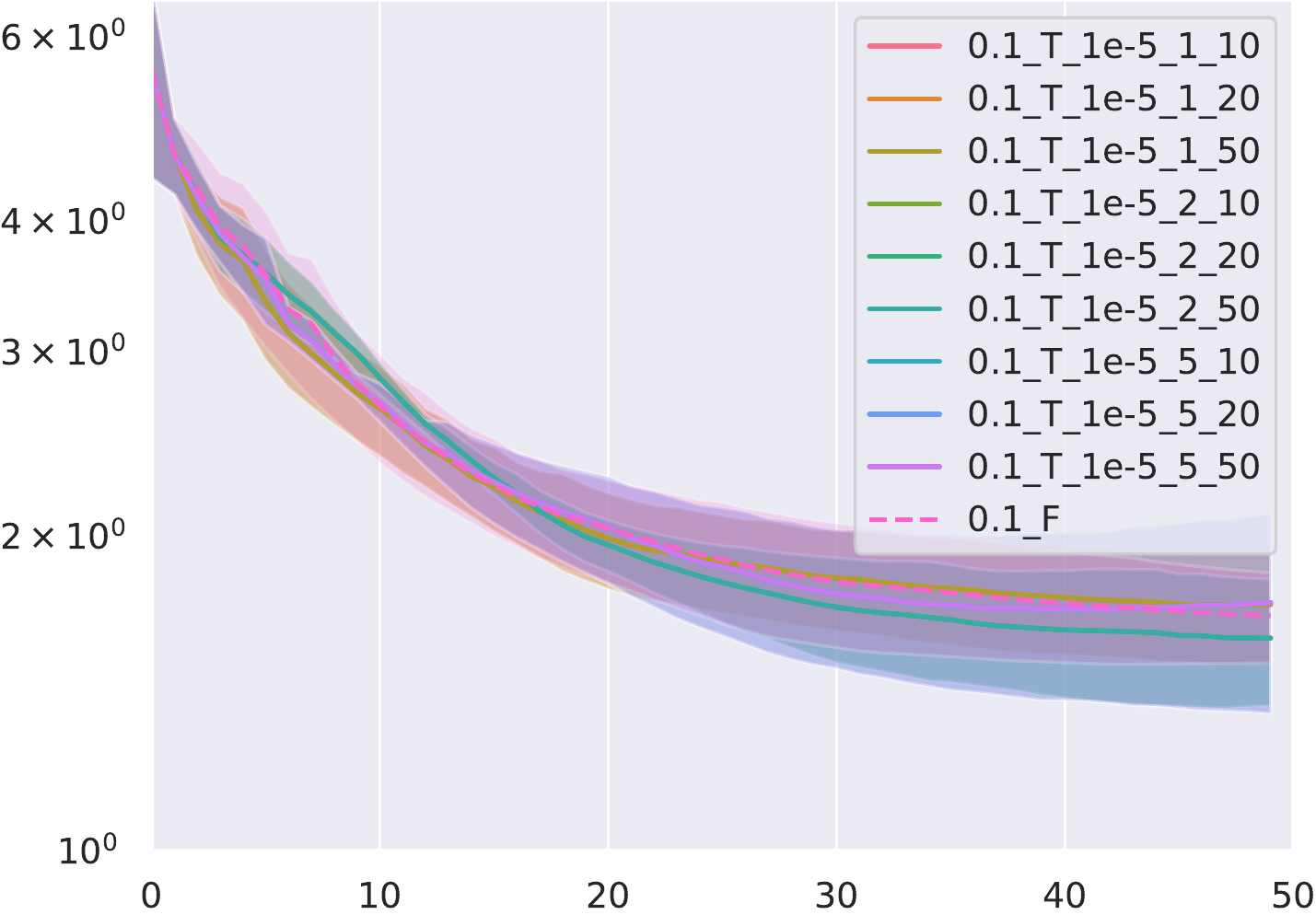}\label{fig:udef_neural_cfr_0_1}}
\subfigure[CFR with $hs=0.5$]{\includegraphics[width=0.325\textwidth]{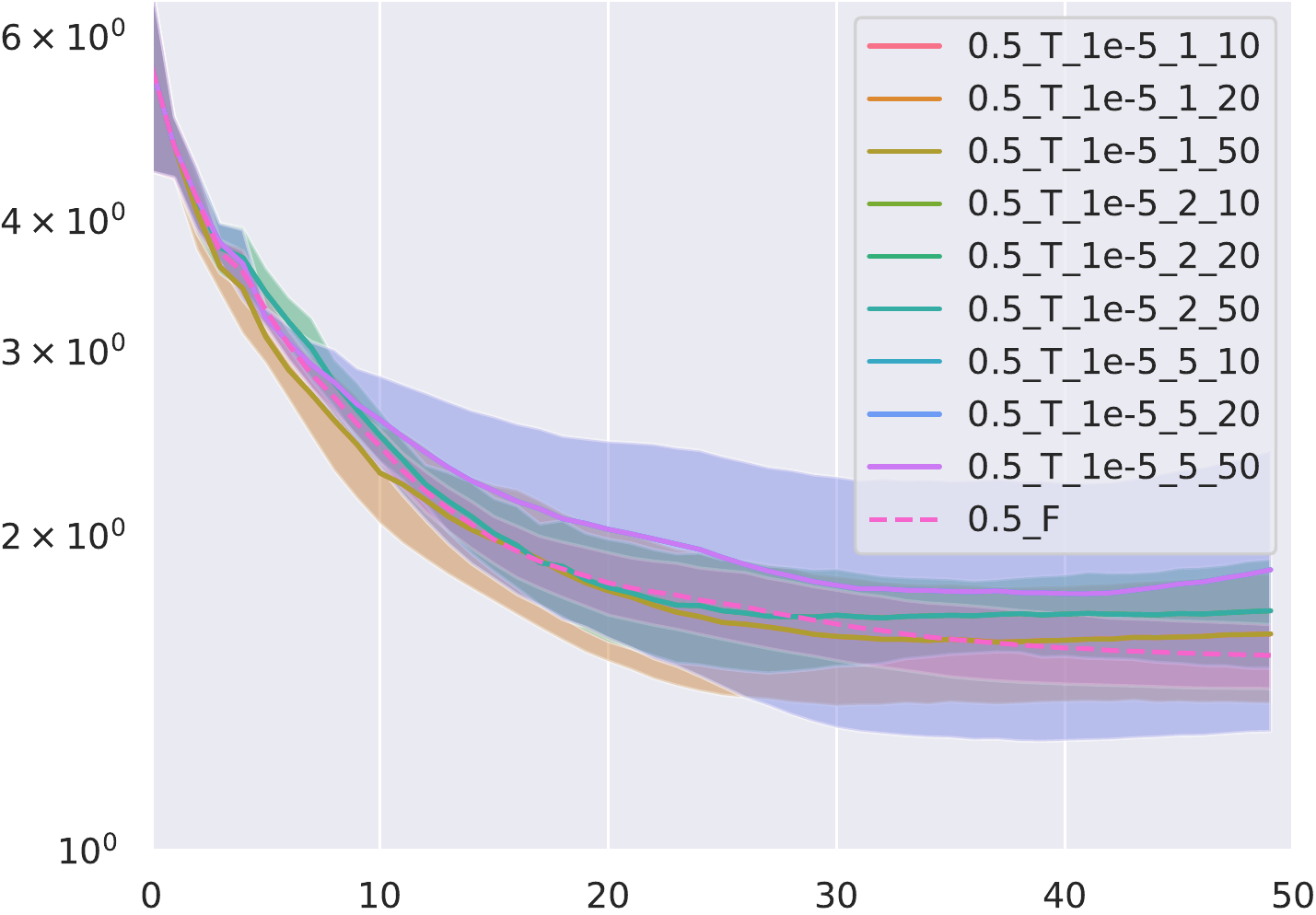}\label{fig:udef_neural_cfr_0_5}}
\subfigure[CFR with $hs=0.9$]{\includegraphics[width=0.325\textwidth]{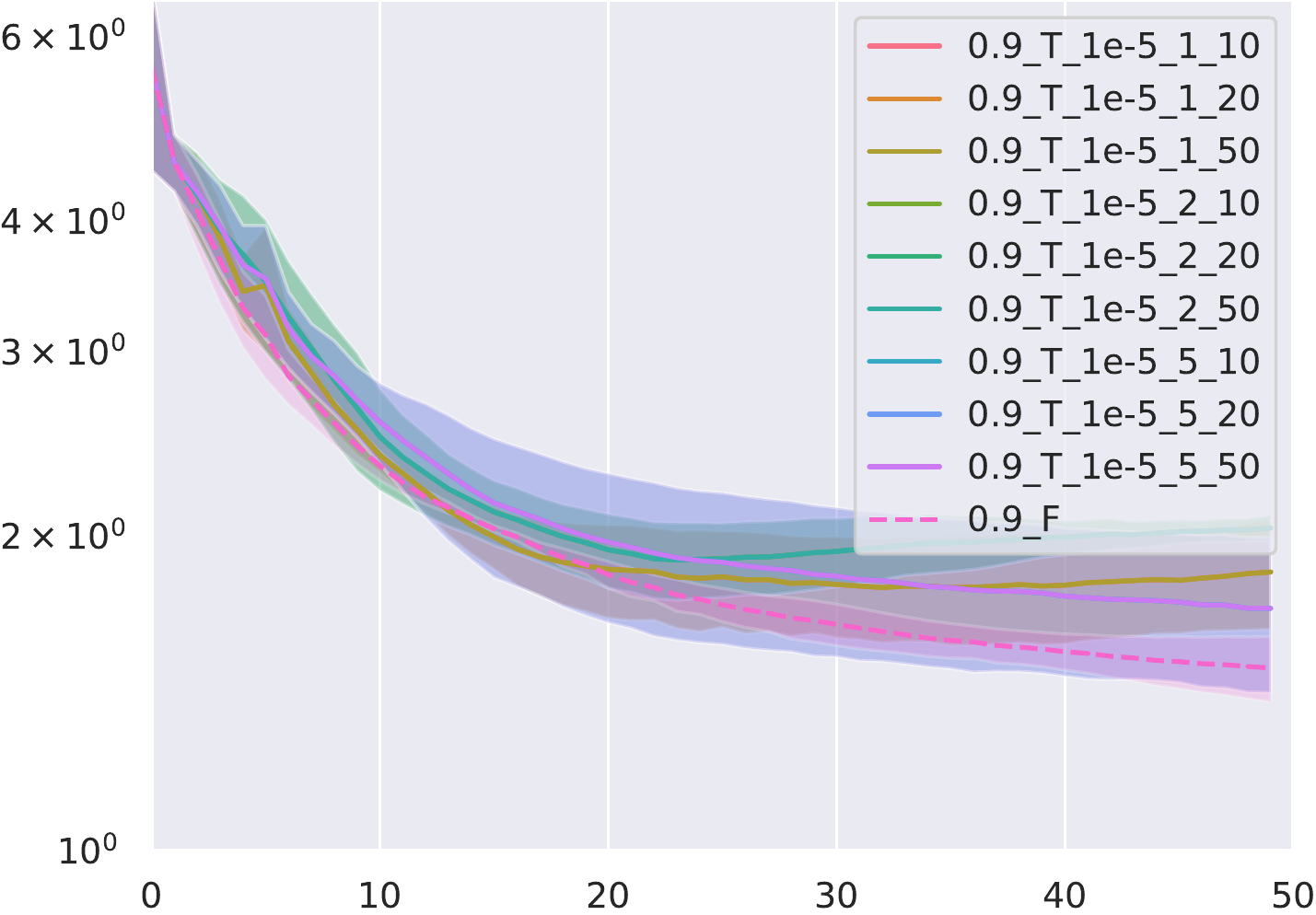}\label{fig:udef_neural_cfr_0_9}}
%\vspace{-9pt}
\caption{Results of \textsc{NashConv} of UDEF with optimizing transform modules and average oracles in NFSP and Deep CFR. The legend in the figure should read as $\langle hs, meta\_train, meta\_lr, meta\_steps, meta\_train\_max\rangle$.}
\label{fig:tran_opt_ao}
%\vspace{-9pt}
\end{figure*}
\subsection{Optimizing Transform Modules Only}
In the first experiment, we optimize the transform modules where the AOs are represented by explicit functions. The results are shown in Figure~\ref{fig:tran_opt}.  We select the NFSP and linear CFR with the neural network representation of the transform modules as our baselines, where linear CFR accelerates CFR by two orders of magnitude~\cite{brown2019solving}. For NFSP (the top row of Figure~\ref{fig:tran_opt}), we set $meta\_lr=10^{-4}$, $meta\_steps\in\{2, 5, 10\}$ and the optimization of the AO and transform modules is fixed during the training as $meta\_train\_max=50$. For $hs=0.1$ and $0.5$, we observe that the training outperforms the baseline for all values of $meta\_steps$. On the contrary, for $hs=0.9$ the training becomes more unstable: setting $meta\_steps=2$ or $5$ results in worse performance and $meta\_steps=10$ may only become competitive with the baseline. For the CFR results (the bottom row of Figure~\ref{fig:tran_opt}), we set $meta\_lr=10^{-5}$ as we observe that all results with $meta\_lr=10^{-4}$ as in NFSP are not stable and worse than the baseline, which implies that CFR is more sensitive to the optimizing of the transform modules than NFSP. The similar results are also observed with varying $meta\_steps\in\{1, 2, 5\}$. When $hs=0.5$ or $0.9$ and $meta\_steps=1$ and $2$, UDEF outperforms the baselines. We note that varying $meta\_train\_max$ does not influence the results much. 

% \clearpage
\subsection{Optimizing Transform Modules and AOs}
In the second experiment, we optimize both transform modules and AOs. We focus on optimizing the LAO, where GAO is represented by an explicit function. The neural network structure of the LAO is Conv1D-based meta solver proposed in~\cite{feng2021discovering}. We randomly sample $10^6$ random games and compute the corresponding LAO values as labels to pretrain the LAO. In each iteration, we use the LAO to generate the values for averaging the responses and optimize LAO simultaneously with transform modules. The experimental results are shown in Figure~\ref{fig:tran_opt_ao}.

The results with NFSP (the top row of Figure~\ref{fig:tran_opt_ao}) suggest that the neural network representation of LAO increases instability during training, compared with the explicit functions, as the training of AO and transform modules influences more parameters of UDEF. Still, we observe that setting $meta\_steps$ properly results in UDEF outperforming the NFSP. Similarly as in previous section, the larger $hs$ gets the more unstable the training becomes. For $hs=0.5$, UDEF consistently outperforms the baseline.

With CFR (the bottom row of Figure~\ref{fig:tran_opt_ao}), UDEF outperforms the baseline when $hs=0.1$. However, for $hs$ equal to $0.5$ or $0.9$, UDEF becomes inferior in all configurations. This indicates that optimizing both AO and the transform modules makes the training extremely unstable. 

Overall, we make three key observations: i) increasing the value of $hs$, i.e., generating new experiences from historical responses, rather than from a new response, makes the algorithms unstable when optimizing AOs and the transform modules, ii) as linear CFR uses linear average as LAO, which puts more weights on later responses, linear CFR becomes more unstable compared to NFSP and may result in instability of training even without the optimization of AO and the transform modules, and iii) UDEF is sensitive to the hyperparameters of training -- more sophisticated hyperparameter selection methods, as well as the optimizing methods, is hence required to increase the representation power of neural networks in UDEF. 
% \begin{figure}[ht]
%     \centering
%     \includegraphics[width=0.45\textwidth]{figures/results/udef.pdf}
%     \caption{NashConv.}
%     \label{fig:pretrain_leduc}
% \end{figure}

% \begin{figure}[ht]
% %\vspace{-9pt}
% \centering
% \subfigure[NFSP]{\includegraphics[width=0.225\textwidth]{figures/results/udef_nfsp.pdf}\label{fig:udef_nfsp}}
% \subfigure[CFR]{\includegraphics[width=0.225\textwidth]{figures/results/udef_cfr.pdf}\label{fig:udef_cfr}}
% %\vspace{-9pt}
% \caption{Results of optimizing transform modules in NFSP and Deep CFR. The legend in the figure is historical sampling+meta\_train+meta\_lr+meta\_train\_steps+meta\_train\_loops}
% \label{fig:tran_opt}
% %\vspace{-9pt}
% \end{figure}

% \subsection{Optimizing Average Oracles Only}

% \begin{figure}[ht]
%     \centering
%     \includegraphics[width=0.225\textwidth]{figures/results/udef_nfsp_neural.pdf}
%     \caption{NFSP with neural AO}
%     \label{fig:my_label}
% \end{figure}
% \begin{figure}
%     \centering
%     \includegraphics[width=0.45\textwidth]{figures/results/nfsp.png}
%     \caption{Caption}
%     \label{fig:my_label}
% \end{figure}

% \subsection{Analysis of Transform Modules}
% \textcolor{red}{What are the optimal transform modules with the known average oracles, e.g., Nash, Linear or $\alpha$-rank?}
% \subsection{Generalizability of Learned Modules}
% In this section, we will only train the response oracle and directly use the average oracle and transform modules from other games. (One common setting is training the average oracle and transform modules on small games and then testing them on large games.)

\section{Conclusion}

In this work, we propose UDEF, a novel unified perspective of DEF, which leverages the powerful expressivity of neural networks to unify the two widely-used frameworks, i.e., PSRO and CFR. The four novel components of UDEF include: i) a novel response oracle, ii) two transform modules, iii) two average oracles, and iv) a novel optimization method to select the components of the two frameworks. Experiments demonstrate the effectiveness of UDEF. We believe that this work unifies and generalizes the two largely independent frameworks and may accelerate research on both frameworks from a unified perspective. 

% \section{Related Works}
% NFSP~\cite{heinrich2016deep},
% PSRO~\cite{lanctot2017unified},
% % XDO~\cite{mcaleer2021xdo}, 
% Deep CFR~\cite{brown2019deep}

\clearpage
\balance
\bibliographystyle{icml2022}
\bibliography{def}
\clearpage
\nobalance
\appendix

\section{Oracles and Modules}
The illustrative figures of different oracles and modules are displayed in Figure~\ref{fig:udef_oracles_and_modules}.
\begin{figure}[H]
\centering
\subfigure[RO]{
\includegraphics[width=0.175\textwidth]{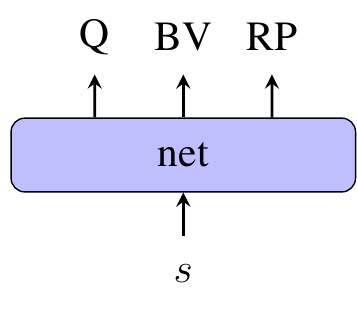}
\label{fig:ro}
}
\subfigure[AO]{
\includegraphics[width=0.175\textwidth]{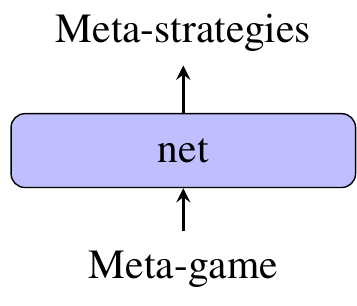}
\label{fig:ao}
}
\subfigure[Pre-tran.]{
\includegraphics[width=0.175\textwidth]{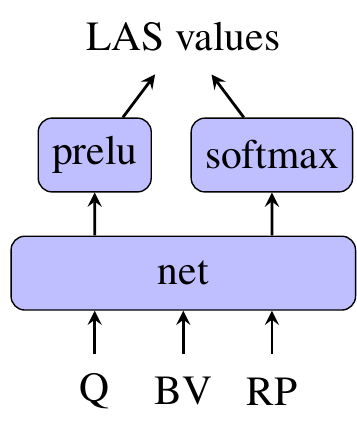}
\label{fig:pre}
}
\subfigure[Post-tran.]{
\includegraphics[width=0.175\textwidth]{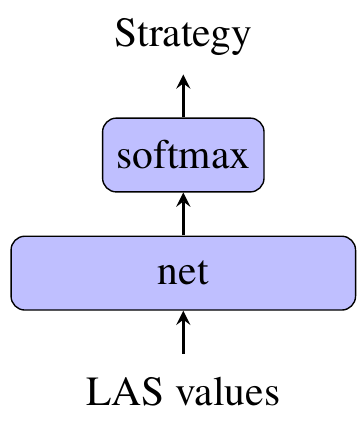}
\label{fig:post}
}
\caption{Design of neural networks in UDEF. We note that both LAO and GAO share the same structure.}
\label{fig:udef_oracles_and_modules}
\end{figure}

\newpage
\section{Details of Pretraining Procedure}

In this section, we present the detailed pseudo codes of different subroutines in UDEF. 
% \subsection{Pretraining}
\begin{algorithm}[ht]
\caption{Pretraining of Average Oracle}
\begin{algorithmic}[1]
\STATE Initialize the AO net's parameters $\psi$
\WHILE{not terminate}
\STATE Determine the number of actions, $k$
\STATE Randomly sample a batch of games, $\bm{X}$
\STATE Compute the solutions of $\bm{X}$ as labels $\bm{y}$
\STATE Using  $\langle\bm{X}, \bm{y}\rangle$ to train $\psi$
\ENDWHILE
\end{algorithmic}
\end{algorithm}

If the dimension of LAS is larger than or equal to the number of actions, we can select a set of neurons to train the transform modules independently. We can also train the two modules jointly by using AO. 
\begin{algorithm}[ht]
\caption{Pretraining of Transform Modules}
\begin{algorithmic}[1]
\STATE Given the pretrained AO $\psi$, the pre-transform and post-transform modules $\phi_{1}$ and $\phi_{2}$
\WHILE{not terminate}
\STATE Randomly sample a batch of Q values and reaching probabilities $\bm{X}$
\STATE Compute the solutions under $\psi$ of $\bm{X}$ as labels $\bm{y}$
\STATE Using  $\langle\bm{X}, \bm{y}\rangle$ to train $\phi_{1}$ and $\phi_{2}$
\ENDWHILE
\end{algorithmic}
\end{algorithm}

\end{document}